\title{Another look at the Lady Tasting Tea and differences between permutation tests and randomization tests}
\author{Jesse Hemerik\footnote{Biometris, Wageningen University \& Research}
\footnote{Address for correspondence: Jesse Hemerik, Biometris, Wageningen University \& Research, P.O. Box 16, 6700 AC Wageningen, The Netherlands. e-mail: jesse.hemerik@wur.nl}
 \phantom{.}and
Jelle Goeman\footnote{Biomedical Data Sciences, Leiden University Medical Center, Einthovenweg 20,
2333 ZC Leiden, The Netherlands}
}

\documentclass[11pt]{article}

\usepackage{amsmath}
\usepackage{amsfonts}
\usepackage{bm}
\usepackage{bbm}
\usepackage{amsthm}
\usepackage{comment}
\usepackage[nottoc]{tocbibind}
\usepackage{graphicx}

\usepackage{subcaption}

\usepackage{changepage}
\usepackage{natbib}
\usepackage{abstract}
\usepackage[colorlinks=true,citecolor=blue,runcolor=blue,linkcolor=blue, pdfborder={0 0 0}]{hyperref}

\usepackage{tikz}

\usepackage{algorithm}
\usepackage{algorithmic}

\usepackage{setspace}

\definecolor{darkblue}{rgb}{0.0, 0.0, 0.55}

\theoremstyle{plain}
\newtheorem{theorem}{Theorem}

\theoremstyle{definition}

\usepackage[margin=3cm]{geometry}

\newcommand{\G}{\mathcal{G}}
\newcommand{\W}{\mathcal{W}}
\newcommand{\Y}{\mathcal{Y}}

\begin{document}
\maketitle
\renewcommand{\abstractname}{Summary}

\begin{abstract}
 
\noindent 
The  statistical literature is known to be inconsistent in the use of the terms ``permutation test'' and ``randomization test''. 
Several authors succesfully argue that these terms should be used to refer to two  distinct classes of tests and that there are major conceptual differences between these classes.
The present paper explains an important  difference in mathematical reasoning between these classes: a permutation test fundamentally requires that the set of permutations has a  group structure, in the algebraic sense; the reasoning behind a randomization test is not based on such a group structure and it is possible to use an experimental design that does not correspond to a group.
In particular, we can  use a randomization scheme where the number of possible treatment patterns is larger than in standard experimental designs. This leads to exact \emph{p}-values of improved resolution, providing increased power for very small significance levels, at the cost of decreased power for larger significance levels.
We discuss applications in randomized trials and elsewhere.
Further, we explain that Fisher's famous Lady Tasting Tea experiment, which is commonly referred to as the first permutation test, is in fact a randomization test. This distinction is important to avoid confusion and invalid tests.
\\
\\
\emph{keywords:} Permutation test;  Group invariance test; Randomization  test; Lady Tasting Tea;

\end{abstract}

\section{Introduction}
The statistical literature is  very inconsistent in the use of the terms ``permutation tests'' and ``randomization tests'' \citep{onghena2018randomization,rosenberger2019randomization}. Both terms are often used to refer to tests that involve permutations. Sometimes these two terms are considered to refer to strictly distinct classes, sometimes to the same and sometimes to partly overlapping classes. The confusion surrounding differences between such tests is an important issue, because there are major differences between permutation tests and randomization tests in the sense of e.g. \citet{onghena2018randomization}, \citet{kempthorne1969behaviour} and \citet{rosenberger2019randomization}, whose definitions  we will follow here. Those authors use these terms to refer to strictly distinct classes of tests and discuss the terms in detail.
Permutation tests are based on random sampling from populations and randomization tests are based on experimental randomization of treatments. With ``randomization of treatments'' we refer to the (physical) randomization that is part of the experimental design. With ``permutation-based tests'' we will refer to all tests involving permutations, including randomization tests involving permutations.

Another point of confusion has been the role of a \emph{group structure}, in the algebraic sense, in permutation-based tests.
\citet{southworth2009properties}, among others, explain that for permutation tests to have proven properties, it is important that the set of permutations used has such a group structure, as we discuss in Section \ref{secp}. 
For example, the set of \emph{balanced permutations}, which is a subset of a permutation group, does not have a group structure, and using it within a permutation test tends to lead to a very anti-conservative test. A balanced permutation, roughly speaking, is a permutation map that moves exactly 50 percent of the cases to the control group and 50 percent of the controls to the case group.
Balanced permutations (not to be confused with stratified permutations) have been used in several publications 
\citep{fan2004normalization,jones2007genome} but \citet{southworth2009properties} warn against their use. 

A common reference in the permutation literature is the  ``Lady Tasting Tea" experiment, decribed in \citet[][Ch. II]{fisher1935}. This experiment is commonly referred to as  the first published permutation test \citep{wald1944statistical, hoeffding1952large, anderson2001permutation,  lehmann2005testing,  langsrud2005rotation, mielke2007permutation,  phipson2010permutation, winkler2014permutation}. Indeed, like permutation tests, this test is based on permutations. We will see, however, that it is not a permutation test, but a randomization test, in the sense of \citet{onghena2018randomization}, \citet{rosenberger2019randomization} and   \citet{kempthorne1969behaviour}.

A main goal of the present paper is to explain the difference between two classes of  tests involving permutations (or other transformations): tests that fundamentally rely on a group structure, and tests that do not, in an appropriate sense. To our knowledge this explicit distinction has not been made before. It is connected to the difference between permutation tests and randomization tests: the former fundamentally rely on a group structure and the latter  do not -- with the caveat that a randomization test should reflect the randomization scheme of the (physical) experiment. If the randomization scheme corresponds to a group, then the test also involves that group; otherwise the test does not involve a group.
The fundamental point of this paper is as follows:  the mathematical reasoning underlying a randomization test is not based on a group structure (even if a group happens to be used);  the reasoning underlying a permutation test, on the other hand, is always based on a  group structure and is completely different from the reasoning underlying randomization tests. 
In the existing literature, many randomization tests involve   a group, but randomization tests that do not involve a group are also often considered \citep{onghena1994randomization,onghena2005customization,rosenberger2019randomization}.
The  further contributions of this paper  are as follows.

First of all, since permutation-based randomization tests do not require a group structure, it can be useful to consider a randomization scheme that does not correspond to a group. 
We introduce the  idea of using an alternative  randomization scheme to increase  the number of possible treatment patterns. This increases the resolution of the \emph{p}-value, thus improving  power for very small significance levels $\alpha$, at the price of  power loss for larger $\alpha$.

In addition, this paper provides the caveat that the Lady Tasting Tea experiment is rather different from  permutation tests \citep[in the sense of e.g.][]{onghena2018randomization}.
Referring to the Lady Tasting Tea experiment as an example of a permutation test, as is often done, can put readers on the wrong foot, since the reasoning underlying this experiment is not based on a group structure.
Referring to the Lady Tasting Tea may have contributed to the confusion that has led researchers to design invalid permutation tests without a group structure \citep{southworth2009properties}. The purpose of this paper is not  to identify the first permutation test, which would not be straightforward \citep{berry2014chronicle}.

This paper is built up as follows. In Section \ref{secp} we review existing results on permutation and group invariance tests, empasizing the key role of the group structure of the permutations.
In Section  \ref{secltt} we discuss the Lady Tasting Tea experiment, emphasizing why it does \emph{not} require a group structure to control the type I error rate.
In Section \ref{secgenct} we generalize the test of Section  \ref{secltt}, providing a  general randomization test and mentioning  applications. In Section \ref{highrescc} we apply the general randomization test in a randomized trial setting, discussing how we can obtain higher-resolution \emph{p}-values than with a canonical permutation-based  test. The performance of our alternative  test is illustrated with simulations in Section \ref{secsim}. We end with a discussion.

\section{Permutation tests and group invariance tests} \label{secp}

In the present section we explain the mathematical reasoning behind permutation tests, focusing on type I error control. 
As mentioned, the terms ``permutation test" and ``randomization test" have been used   inconsistently in the literature.
A typical example of a permutation test in the sense of  \citet{onghena2018randomization}, \citet{kempthorne1969behaviour} and \citet{rosenberger2019randomization}  is discussed in \citet[][pp. 58-59]{fisher1936coefficient}. In this thought experiment, measurements of the statures of 100 Englishmen and 100 Frenchmen are considered. 
These observations are assumed to be randomly sampled from their respective populations. Such a  model, where observations are randomly sampled from their populations, is typical for permutation tests in the sense of  for example \citet{kempthorne1969behaviour}, \citet{onghena2018randomization} and \citet{rosenberger2019randomization}. 
Note that in this example, there is no randomization of treatments as in, for example,  clinical trials.
In the example in \citet[][pp. 58-59]{fisher1936coefficient}, to test whether ``the two populations are homogeneous", the difference between the two sample means is computed and this is repeated for each permutation of the 200 observations.
The null hypothesis is rejected if the original difference is larger than most of the differences obtained after permutation.
We will return to this example below.

Permutation tests   are special cases of the general group invariance test. 
The definition of the group invariance test in, for example, \citet{hoeffding1952large}, \citet{lehmann2005testing} and \citet{hemerik2018exact}   is rather general, so that  many  randomization tests also fall under it.
The relationships between permutation tests, randomization tests and group invariance tests are illustrated in Figure \ref{venndiagram}.
The principle underlying the group invariance test can also be used to prove properties of various permutation-based multiple testing methods \citep{westfall1993resampling,tusher2001significance,meinshausen2005lower,hemerik2018false,hemerik2019permutation}.

\begin{figure}  
\centering    
        \begin{tikzpicture}[rounded corners] 

      \draw (0,-0.2) rectangle +(7.4,3.7);
      \draw (3.3,3.75) node[black, text width= 5cm]{Group invariance tests ($\mathbf{G}$)};
      \draw (0.2,0.7) rectangle +(4.2,2.5);
      \draw (2.3,2.1) node[black, text width= 4cm]{Permutation tests ($\mathbf{P}$)};
      \draw (4.8,0.7) rectangle +(4.2,2.5);
       \draw (11.5,2.1) node[black, text width= 4.8cm]{Randomization tests ($\mathbf{R}$)};
       \draw (7.6,1.3) node[black, text width= 4cm]{$\mathbf{R}\cap\mathbf{G}$};
       \draw (9.7,1.3) node[black, text width= 4cm]{$\mathbf{R}\setminus\mathbf{G}$};
       \draw (7.1,0.2) node[black, text width= 4cm]{$\mathbf{G}\setminus(\mathbf{P}\cup\mathbf{R})$};
    \end{tikzpicture}
    \caption{A Venn diagram showing relationships between permutation tests ($\mathbf{P}$), randomization tests ($\mathbf{R}$) and group invariance tests ($\mathbf{G}$). 
Permutation tests ($\mathbf{P}$) are a subclass of group invariance tests ($\mathbf{G}$) and are based on random sampling from populations. Randomization tests ($\mathbf{R}$) are based on randomization of treatments.  The randomization scheme sometimes corresponds to a group ($\mathbf{R}\cap\mathbf{G}$) and sometimes does not ($\mathbf{R}\setminus\mathbf{G}$).  If the scheme corresponds to a group then this is often a set of permutations, but not always.
    There are also group invariance tests that fall outside all these categories ($\mathbf{G}\setminus(\mathbf{P}\cup\mathbf{R})$).
    An example of a test from class $\mathbf{P}$ is Fisher's thought experiment with Englishmen and Frenchmen described in Section \ref{secp}.  An example from the class $\mathbf{R}\cap\mathbf{G}$ is Fisher's Lady Tasting Tea experiment.
  A different example, based on sign-flipping instead of permutation, is the test in \citet[][\S21]{fisher1935}. 
     An example from the class $\mathbf{R}\setminus\mathbf{G}$ is in Section \ref{highrescc}. An example from the class $\mathbf{G}\setminus(\mathbf{P}\cup\mathbf{R})$ are certain tests based  on rotations \citep{solari2014rotation} or sign-flipping \citep{winkler2014permutation}.} \label{venndiagram}
\end{figure}
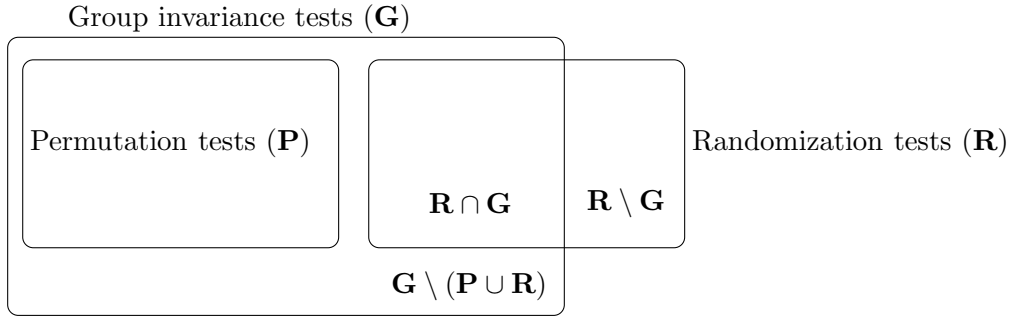 

\color{black}
A general definition of a group invariance test is as follows.
Generalizations of this framework, such as two-sided tests, are possible.
Let $X$ model random data with support $\mathcal{X}$. For example, $X$ could be a random vector or matrix. Consider a set $\G$ of permutation maps or other transformations $g:\mathcal{X}\rightarrow \mathcal{X}$.  We will assume that $\G$ is finite, although generalizations are possible. The set $\G$ is assumed to have a group structure with respect to the operation of composition of maps, which means that:  $\G$ contains the identity map $x\mapsto x$; every element in $\G$ has an inverse; and for all $g$, $h\in \G$, $g\circ h\in \G$ \citep{hoeffding1952large}.
Further, we consider some test statistic $T:\mathcal{X}\rightarrow \mathbb{R}$.
Consider a null hypothesis  $H_0$ that implies that the joint distribution of all test statistics $T(g(X))$ with $g\in \G$ is invariant under all transformations  of $X$ in $\G$ \citep{hemerik2018exact}. This holds in particular if the data $X$ are themselves transformation-invariant, i.e.,  if
\begin{equation} \label{Ginvariance}
g(X)\,{\buildrel d \over =}\,X  
\end{equation}
 for every $g\in \G$.

A typical example of such a setting is the thought experiment from \citet[][pp. 58-59]{fisher1936coefficient}, mentioned above.
Let $X_1,...,X_{100}$ be the statures of the Englishmen and let  $X_{101},...,X_{200}$ be the statures of the Frenchmen.
The test statistic considered in  \citet[][pp. 58-59]{fisher1936coefficient} is
\begin{equation}
T(X)=\Big|\frac{1}{100}\sum_{i=1}^{100} X_i - \frac{1}{100} \sum_{i=101}^{200} X_i\Big|.   \label{tst}
\end{equation} 
The  null hypothesis $H_0$ is that $X_1,...,X_{200}$ are i.i.d..  The null hypothesis is rejected if the original test statistic is larger than most of the statistics obtained after permutation.  
(In practice one would nowadays sample  random permutations, to avoid having to use all possible permutations, see \citealp{hemerik2018exact}.)
The group $\G$ that Fisher considers consists of all permutation maps
 $g: \mathbb{R}^{200} \rightarrow \mathbb{R}^{200}$. Here, every $g\in \G$ is of the form 
 $$(x_1,...,x_{200})\mapsto (x_{\pi_1},..., x_{\pi_{200}}  ),$$
  where $(\pi_1,...,\pi_{200})$ is a permutation of $(1,...,200)$.
 Note that  $X$ is then $\G$-invariant under $H_0$, i.e.,   \eqref{Ginvariance} holds for every $g\in \G$.
 
 As another example, consider random data $X$ with support in $\mathbb{R}^{n}$, with independent entries that are symmetric about their means. Suppose that under $H_0$, the entries have mean 0. This may for example be the case if each entry of $X$ is the difference of two paired observations. 
 Then the distribution of $X$ is invariant under all transformations in $\G$ under $H_0$ if we define $\G$ to be the group of all sign-flipping maps of the form 
\begin{equation} \label{signflip}
 (x_1,...,x_n)\mapsto (s_1x_1,...,s_nx_n), 
\end{equation}
 with $(s_1,...,s_n)\in \{-1,1\}^n$. We may take $T((x_1,...,x_n))=\sum_{i=1}^n x_i$.
 If the data generating mechanism involves randomization of treatments (rather than sampling from populations), such a test can be considered a randomization test.
The test already appears in \citet[][\S21]{fisher1935}, albeit without explicit proof. See also \citet{basu1980randomization}.
 
In  both examples above, we can apply the general group invariance test to test $H_0$. This test already appears in the literature \citep{hoeffding1952large,lehmann2005testing, hemerik2018exact}, but for completeness we include the result and its proof.  

We will write $gX=g(X)$ for short. Let $k=\lceil (1-\alpha)|\G| \rceil$, the smallest integer which is larger than or equal to $(1-\alpha)|\G|$. Let 
$T^{(1)}(X)\leq  T^{(2)}(X)    \leq ... \leq  T^{(k)}(X)    \leq  ... \leq T^{(|\G|)}(X)$
 be the sorted values $T(gX)$ with  $g\in \G$.

 \begin{theorem} \label{basic}
The size of the group invariance test is at most $\alpha$, i.e., under $H_0$, $\mathbb{P}\big\{T(X)>T^{(k)}(X)\big\}\leq \alpha.$
\end{theorem}
\begin{proof}
By the group structure, $\G g=\G$ for all $g\in \G$. Hence $T^{(k)}(gX)=T^{(k)}(X)$ for all $g\in \G$.
Let $G$ have the uniform distribution on $\G$. Then under $H_0$, the rejection probability is
$$
\begin{aligned}
&\mathbb{P}\big\{T(X)>T^{(k)}(X)\big\}=\\
&\mathbb{P}\big\{T(GX)>T^{(k)}(GX)\big\}=\\
&\mathbb{P}\big\{T(GX)>T^{(k)}(X)\big\}.
\end{aligned}
$$
The first equality follows from the null hypothesis and the second equality holds since $T^{(k)}(X)=T^{(k)}(GX)$. 
Since $G$ is uniform on $\G$, the above probability equals
$$\mathbb{E} \Big [|\G|^{-1} \cdot \big|\big \{g\in \G: T(gX)> T^{(k)}(X) \big \}\big| \Big ] \leq \alpha,$$ 
as was to be shown.
\end{proof}

Under additional assumptions, the test is exact, i.e., the rejection probability is exactly $\alpha$ under $H_0$ \citep{hemerik2018exact}.
In the above proof we used the group structure, which guarantees the symmetry property $\G g=\G$ for all $g\in \G$.  A different proof, based on conditioning on the pooled sample, is also possible and also requires using this symmetry property \citep[first proof of Theorem 1 in][]{hemerik2018exact}.

Write $\G X=\{gX:g\in \G\}$ and assume for convenience that $gX$ and $g'X$ are distinct with probability 1 if  $g,g'\in \G$ are distinct. This is usually the case if $X$ is continous.
The permutation  test is based on the fact that under $H_0$, for every permutation $g\in \G$ the probability  $\mathbb{P}\{T(gX)>T^{(k)}(X)\}$ is the same. 
The reason is that under $H_0$, for every $g\in \G$, the joint distribution of $(gX,\G X)$ is the same.
This is  because if $g$, $g'\in \G$, under $H_0$ we have  $$(gX,\G X)= (gX,\G gX)   \,{\buildrel d \over =}\,  (X,\G X) \,{\buildrel d \over =}\, (g'X,\G g'X) = (g'X,\G X). $$
When $\G g=\G$ does not hold for all $g\in \G$, then the above does not generally hold under $H_0$.

The group structure of $\G$ implies  that $\G g=\G$ for all $g\in \G$.  Under the mild condition that  all $g\in \G$ are surjective, the reverse implication also holds, i.e., if $\G g=G$ for all $g\in \G$, then $\G$ is a group. For example, if $\G g=\G$ for all $g\in \G$, there are   $h,g\in \G$ with $hg=g$. It follows that $\G$ contains an identity element, and the other group properties also easily follow. We conclude that in the argument underlying the permutation test, the group structure is key. 

In practice it is often computationally infeasible to use a test based on the full group $\G$ of transformations.  Researchers then usually resort to using a limited number of random transformations, uniformly sampled from the group $\G$. It is then still possible to obtain an exact test (see \citealp{hemerik2018exact} and \citealp{phipson2010permutation} for a detailed treatment).

\section{The Lady Tasting Tea and randomization tests}

Unlike a permutation test, a randomization test is based on data collected in an experiment involving randomization of treatments. The  randomization scheme of the physical experiment does not necessarily correspond to a group, and if it does not, the statistical test does not involve a group either.
Even if a group is used, the reasoning underlying a randomization test is not based on the group structure and is very different from the reasoning underlying permutation tests.

In this section, we first discuss the Lady Tasting Tea experiment, explaining that the reasoning underlying the test is not based on a group structure, because it is based on randomization.
This experiment is a special case of a general randomization test that we discuss in Section \ref{secgenct}. In Section \ref{highrescc}  we apply this test to provide higher-resolution 
 \emph{p}-values in randomized trials.

\subsection{The Lady Tasting Tea experiment}  \label{secltt}
In the Lady Tasting Tea experiment \citep[][Ch. II]{fisher1935},  the null  hypothesis  is that a particular lady cannot distinguish between two types of cups of tea with milk: cups in which the tea was added first and cups in which the milk was added first. To test the null hypothesis, which we  denote by $H_0$, the experimenter ``mixes eight cups of tea, four in one way and four in the other," and presents them ``to the subject for judgment in a random order." The experimental setup is  made known to the lady. 
The lady then  tastes from the cups and has to determine which four cups in the sequence of eight had milk added first. 
Fisher actually performed the  experiment \citep{box1978ra, berry2014chronicle}.

There are $\binom{8}{4}=70$ possible orders, with respect to the two types of cups.
Suppose $H_0$ is true. If the lady guesses every pattern with probability $1/70$, then the probability that she chooses the correct order is $1/70$. Even if she has an a priori preference for a certain order, the probability of guessing correct is $1/70$. Indeed, it is assumed that   the researcher randomizes the true pattern, i.e., he chooses each pattern with equal probability.
Thus, if we reject $H_0$ when the lady identifies all four ``milk first" cups correctly, then the probability of a type I error is $1/70$ ($1.4\%$). The probability that she labels three of the ``milk first" cups correctly is
$\binom{4}{3} \binom{4}{1}/70=16/70$ ($22.9\%$) and the probability of two correct picks is $36/70$ ($51.4\%$). Thus, for example,  when we reject $H_0$ if at least three picks are correct, the level is $16/70+1/70=17/70$ ($24.3\%$). The test is equivalent to an instance of ``Fisher's exact test" \citep{yates1934contingency, fisher1935logic, berry2014chronicle} with pre-fixed marginal frequencies in the $2\times 2$ table. Fisher's exact test, however, was not originally motivated from a permutation or randomization testing perspective. 


Mathematically, we can describe the experiment as follows.
Let $\W\subset\{0,1\}^8$ be the set of vectors containing four 0's and four 1's, so that the cardinality of $\W$ is $R:=|\W|=70$.
Let  the decision of the lady be denoted by $Y$ and let $W$ denote the true order, i.e., the random decision by the experimenter. Here, $Y$ and $W$ are random variables taking values in $\W$. Note that $W$ represents the `treatments' given by the experimenter  and $Y$ respresents the lady's `responses'.
The experimenter's order $W$ is assumed to be uniformly distributed on $\W$. The null hypothesis is 
$$
H_0: \text{ } Y \text{ is independent of } W.
$$
Let $\alpha\in(0,1)$ be the  desired type I error rate. If $\alpha\in A=\{1/70,17/70,53/70,69/70\}$, then $\alpha$ is called \emph{attainable} in the Lady Tasting Tea experiment, meaning that we obtain a test of exactly level $\alpha$ \citep{pesarin2015some}. If $\alpha$ is not attainable, then we obain a test with level strictly less than $\alpha$.

Let $T:  \W\times \W \rightarrow \mathbb{R}$ be a test statistic such that high values of $T(w,y)$ indicate that the patterns $w$ and $y$ are similar, i.e., that there is evidence against $H_0$. 
Let $$T^{(1)}(Y) \leq ... \leq T^{(70)}(Y)$$ be the sorted  statistics $T(w,Y)$ with $w\in \W$. 
Whether the vector of sorted  statistics $(T^{(1)},...,T^{(70)})$ actually depends on $Y$ or not, depends on the definition of $T$; in \citet{fisher1935}, the test statistic is 
\begin{equation} \label{Tfisher}
T(W,Y)=\sum_{i=1}^8 \mathbbm{1}_{\{W_i=1\}\cap \{Y_i=1\} } = \sum_{i=1}^8 W_i Y_i
\end{equation}
and it can be seen the sorted statistics do not depend on $Y$, since $Y$ always has four entries that are $1$.
Let $\lceil (1-\alpha)R \rceil $ be the smallest integer that is at least $(1-\alpha)R$. 
We have the following result \citep{fisher1935}.

\begin{theorem} \label{ltttest}
The test that rejects $H_0$ if and only if 
$   T(W,Y)> T^{(\lceil (1-\alpha)R \rceil)}  $ has size as most $\alpha$.
\end{theorem}
\begin{proof}
Assume  $H_0$ holds.
Conditional on $Y$, $W$ is uniformly distributed on $\W$ and $T^{(\lceil (1-\alpha)R \rceil) }(Y)$ is known. Hence, conditional on $Y$,   the rejection probability is
$$\mathbb{P}\big( W\in \{ w\in\W:     T(w,Y)> T^{(\lceil (1-\alpha)R \rceil)}(Y)   \}  \big )=$$  
$$\frac{1}{R}|\{  w\  \in\W:     T(w,Y)> T^{(\lceil (1-\alpha)R \rceil)}(Y) \}  |\leq \alpha.$$
Thus, marginal over  $Y$, the rejection probability is also at most $\alpha$.
\end{proof}

Observe that when we use the test statistic \eqref{Tfisher}, then taking $\alpha\in A$ indeed results in an exact test.  This follows from the fact that
$$
T^{(1)}< T^{(2)}=T^{(3)}=...=   T^{(17)}<   T^{(18)}=T^{(19)}=...
=   T^{(53)} <  T^{(54)} =...=  T^{(69)}  < T^{(70)},
$$
by the argument at the beginning of this section.
If $\alpha\in(0,1)\setminus A$, the level is strictly smaller than $\alpha$.
If the experimenter does not choose randomly from all 70 possible patterns, but uses some smaller set of patterns for him and the lady to choose from, then there may not  be any $\alpha\in(0,1)$ for which the test is exact, since the sorted test statistics may depend on $Y$. This is one of the reasons why using the full set of patterns, in combination with a suitable test statistic $T$, is useful.
However, to prove  Theorem \ref{ltttest}, we did not need to use the group structure of the permutations. The reason is that in the Lady Tasting Tea experiment, under $H_0$ the randomization $W$ of the researcher is by design independent of the reference set $\{(w,Y): w\in \W\}$. Further considerations follow below.

\subsection{A general randomization test} \label{secgenct}

Theorem \ref{ltttest} still applies if the researcher uses a set of permutations that does not correspond to a group. Suppose for example that the researcher omits one of the patterns, thus picking randomly from some set $\W$ of 69 patterns, with or without the lady's knowledge. Denote the set that the lady chooses from by $\Y$. Then $W$ and  $Y$ will still be independent and Theorem \ref{ltttest} still applies if we let  $R=|\W|=69$ and let $$T^{(1)}(W) \leq ... \leq T^{(69)}(W)$$ be the sorted test statistics $T(w,Y)$, $w\in \W$.
Indeed, conditional on $Y$, $W$ will  have a uniform distribution on $\W$.
In fact, we have the following very general randomization test, of which the Lady Tasting Tea experiment is a special case.
We refer to this  result as a randomization test since in most applications of the theorem, the variable $W$ will  represent experimental randomization of treatments \citep{kempthorne1969behaviour, onghena2018randomization}. 
The idea of the theorem  is certainly not new; it is at least  implicitly present in earlier works \citep{morgan2012rerandomization}.

\begin{theorem}[General randomization test] \label{dtest}
Let $\W$ and  $\Y$ be  nonempty sets, where $\W$ is assumed to be finite.  Write $R=|\W|$.
Let $Y$ be a variable taking values in $\Y$ and assume $W$ is uniformly distributed on $\W$. (Here $Y$ and $W$ are variables in a general sense, e.g. they may be random vectors.)
Let $T: \W\times \Y\rightarrow \mathbb{R}$ be some test statistic. 
Consider a null hypothesis $H_0$ that implies that $Y$ is independent of $W$.
Let  $T^{(1)}(Y) \leq ... \leq T^{(R)}(Y)$ be the sorted values $T(w,Y)$ with $w\in\W$.  
Consider the test that rejects $H_0$ if and only if 
$   T(W,Y)> T^{(\lceil (1-\alpha)R \rceil)} (Y) $.
Then the result of Theorem \ref{ltttest} still applies, i.e., the  test has size at most $\alpha$.
\end{theorem}

The proof is analogous to that of Theorem \ref{ltttest}: under $H_0$,
conditional on $Y$, $W$ is uniformly distributed on  $\W$ and $T^{(\lceil (1-\alpha)R \rceil) }(Y)$ is known. Hence, conditional on $Y$,   the rejection probability is
$$\mathbb{P}\big( W\in \{ w\in \W:     T(w,Y)> T^{(\lceil (1-\alpha)R \rceil)}(Y)   \}  \big ) \leq \alpha$$ as before.

Under additional assumptions, the test is exact, i.e., the rejection probability is exactly $\alpha$ under $H_0$.
We assumed that $\W$ is finite, but generalizations to infinite $\W$  are possible, as well as generalizations to non-uniform $W$. We can also define a two-sided test. Moreover, as in Section \ref{secp}, under straightforward additional assumptions, we can prove that the test of Theorem \ref{dtest} is exact for certain $\alpha$, i.e., that the rejection probability is exactly $\alpha$ under $H_0$ \citep{hemerik2018exact}.

Note that in Theorem \ref{dtest}, $Y$ might be a constant, conditional on $W$. In fact, the randomization testing literature often views the outcomes as non-random, conditional on the treatments.  This corresponds to the fact that  randomization tests can be used without an assumption that  the responses are  randomly sampled from populations \citep{cox2009randomization, onghena2018randomization, rosenberger2019randomization}. We discuss this further in the context of randomized trials in Section \ref{highrescc}.

The  general randomization test of Theorem \ref{dtest} has many applications. 
Examples are agricultural experiments and randomized clinical trials. Randomized trials will be discussed in Section \ref{highrescc}. We mention a few other interesting applications here.

First of all, Theorem \ref{dtest} has implications for the  Lady Tasting Tea experiment. In Section \ref{secltt},
 it is assumed that the lady knows beforehand that there are $m$ cups of each type, where $2m$ is the total number of cups she receives.
If for some  reason she does not know that,  then she might label e.g. $m+1$ of the $2m$ items with the same label. In other words, she might pick a pattern from a set containing more patterns than the experimenter chooses from. Theorem \ref{dtest} then says that the type I error probability will nevertheless be at most $\alpha$ under $H_0$. Indeed, in Theorem \ref{dtest}, $\Y$ is allowed to be any nonempty set, so in particular it can be larger than $\W$. 

A further application of  Theorem \ref{dtest} are general sensory tests, of which the Lady Tasting Tea experiment is an example. It is interesting to note that in the  literature on sensory tests, Fisher's experiment has been regarded a ``forerunner of modern sensory analysis" \citep{bi2015revisiting}. For example,  \citet{harris1949measurement} perform a sensory experiment that is analogous to the Lady Tasting Tea experiment, as follows: ``The glasses are arranged at random. The subject is told that four of  them contain the substance and four contain water, and he is asked to taste them all and to separate them into the two groups of four.''

Another application of  Theorem \ref{dtest} are  existing permutation-based randomization tests that are used to evaluate whether some classification algorithm has any predictive ability (both in-sample and out-of-sample).  Such tests can be used to evaluate algorithms for, for example, 
text categorization, fraud detection, optical character recognition and medical diagnosis. Tests of this type are discussed in, for instance, \citet{golland2005permutation}, \citet{airola2010applying}, \citet{ojala2010permutation}, \citet{schreiber2013statistical} and \citet{rosenblatt2019better}.

\subsection{Randomization testing without a group structure: higher-resolution \emph{p}-values} \label{highrescc} 

In randomized trials, often we are interested in comparing two different treatments, for example a drug and a placebo. In such a setting, there is a treatment assignment randomized by the experimenter. In that case, we can use the randomization test of Theorem \ref{dtest}, as explained below. As discussed, we then do not require a group structure, to control the type I error rate. 
We now discuss such a setting in detail.
The tests considered here  will also be studied with simulations in Section \ref{secsim}.

Let  $n\geq 2$ be an integer, assumed even for convenience, and suppose we have  $n$ subjects, $n/2$ of which receive one treatment and $n/2$ of which receive the other treatment. Let $W=(W_1,..,W_n)$ denote the treatments and $Y=(Y_1,...,Y_n)$ the responses, taking values in $\mathbb{R}^n$.   The treatment pattern $W$ is uniformly sampled from a set $ \mathcal{W}\subseteq \{0,1\}^n$.
The most  common type of randomized trial is the \emph{forced balance procedure},  where
\begin{equation} \label{refset}
 \mathcal{W}= \{w\in \{0,1\}^n: \text{ } w \text{ contains } n/2 \text{ 1's}\}
\end{equation}
\citep{rosenberger2015randomization,lachin1988statistical,braun2001optimal}. 
For each $1\leq i \leq n$, the response $Y_i\in \mathbb{R}$ is independent of all the other subjects' treatments and responses. 
We consider the  null hypothesis $H_0$ that $Y$ is independent of $W$.

These assumptions are still rather general. It can be useful to consider a more specific randomization model as in \citet[][\S7]{pitman1937significance}, who assumes an additive treatment effect. 
An important  property of randomization models,  is that to test whether the treatment has an effect on our particular patients,   we do not  need to assume that they are random draws from  populations. We could consider the patients as fixed and $Y$ as constant, conditional on $W$ \citep[][\S7]{pitman1937significance}. Indeed, ``Any assumption that the units are, say, a random sample from a population of units [...] is additional to the specification" of the model \citep{cox2009randomization}. 
This property is discussed in detail in  \citet{onghena2018randomization} and  \citet{rosenberger2019randomization}.

We can invoke Theorem \ref{dtest} to obtain a test that controls the type I error rate.  
We can also obtain an exact test, i.e., a test that rejects with probability exactly $\alpha$ under $H_0$. Consider  the test statistic $T: \mathcal{W}\times \mathbb{R}^n\rightarrow \mathbb{R}$ defined as
 \begin{equation} \label{te1}
T(W,Y)=\sum_{\{i: W_i=1\}}Y_i - \sum_{\{i: W_i=0\}}Y_i.
\end{equation}
Recall that $Y$ may be viewed as random or constant, conditional on $W$. In either case, assume that $Y$ is such that (with probability $1$), for all distinct $w_1, w_2\in \mathcal{W}$,  $T(w_1,Y)\neq T(w_2,Y)$. This is satisfied in particular if  $Y_1,...,Y_n$ have continuous distributions. Write $N=\binom{n}{n/2}.$
The test is exact if $\alpha\in(0,1)$  is a multiple of $1/|\mathcal{W}|$, where $|\mathcal{W}|$ equals 
$N.$
An exact \emph{p}-value is
\begin{equation} \label{pvrtest}
p(W,Y)=\frac{|\{w\in \mathcal{W}: T(w,Y)\geq T(W,Y) \}|  }{ |\mathcal{W}|},
\end{equation}
i.e., if $\alpha\in(0,1)$ is  a multiple of $1/|\mathcal{W}|$, then $\mathbb{P}(p\leq \alpha)=\alpha$ under $H_0$. 
A two-sided exact test can be obtained analogously.

Since Theorem \ref{dtest} applies,  the test essentially does not rely on a group structure. Hence, we may consider taking $\W$ to be a set that does not correspond to a group. 
In a different context, this is also done in \citet{onghena1994randomization} and \citet{onghena2005customization}, where a set of permutations is used that is strictly smaller than the full set of permutations. 
This is done to avoid too repetitive treatment patterns such as ABBBBAAA.
In our setting, if $n=8$, instead of taking $\mathcal{W}$ to be the set of all permutations of $(0,0,0,0,1,1,1,1)$ we could take $\W$ to be a  subset  that does not correspond to a group, and still obtain an exact test (for certain $\alpha$). 
As  \citet{onghena1994randomization} illustrates, this may be useful in some settings.  However, in a typical randomized trial there is no evident reason to only use a subset of the permutations, except to limit the number of permutations for computational reasons.

A  more interesting alternative, when running the experiment, is to draw $w$ from a set that is strictly larger than the set in \eqref{refset}, for example, from the set of all possible labelings, $\{0,1\}^n$ (a \emph{Bernoulli trial}, \citealp{imbens2015causal}). 
Indeed, if the standard  randomization scheme is used, i.e., the forced balance trial, then the smallest possible \emph{p}-value that can be obtained is $1/N$, due to the discreteness of the \emph{p}-value. If $n=8$ for example, then $1/N=1/70$. This means that if the significance level is $\alpha=0.01$ for instance, we have a power of 0 to reject $H_0$. 
Such small $\alpha$ are often used nowadays, for example due to multiple testing. The discreteness of the permutation \emph{p}-value is a well-known downside of permutation-based tests \citep{berger2000pros}.
If we take $\mathcal{W}=\{0,1\}^n$, however, then $|\mathcal{W}|=2^8$, so that the smallest possible \emph{p}-value is $1/2^8=1/256$. If $1/256\leq \alpha<1/70$, this means a uniform improvement in power over the standard randomization test.
Note that there only is a power improvement for very small $\alpha$; for larger $\alpha$ the Bernoulli trial has less power than the forced balance trial.

Under $H_0$, if $\alpha$ is a multiple of $1/2^n$, the test with $\mathcal{W}=\{0,1\}^n$ rejects with probability exactly $\alpha$. Otherwise the test rejects with probability less than $\alpha$ under $H_0$.
For  $\mathcal{W}=\{0,1\}^n$,  to our knowledge it is not known what the optimal choice of $T$ is for testing an additive treatment effect.
In Section \ref{secsim} we will take
 \begin{equation} \label{te2}
T(W,Y)=\sum_{\{i: W_i=1\}}(Y_i-\overline{Y}) - \sum_{\{i: W_i=0\}}(Y_i-\overline{Y}),
\end{equation}
where $\overline{Y}=n^{-1}(Y_1+...+Y_n)$.
Using this  test statistic ensures that under $H_0$, the expected value of $T(W,Y)$ does not depend on the random labelling $W$. 
In Appendix \ref{Acovb} we show how the Bernoulli trial can be modified to enforce covariate balancing.

That it is possible to take $\mathcal{W}=\{0,1\}^n$ has been noted by several authors \citep{pocock1979allocation,kalish1985treatment,lachin1988properties,wei1988properties,suresh2011overview,imbens2015causal,rosenberger2019randomization}. They do not recommend this approach, but merely mention it as a possibility, while focusing on more common randomization schemes. Their main argument against taking $\mathcal{W}=\{0,1\}^n$ is that it leads to less power than the usual approach of restricted randomization  \citep[][p.188]{pocock1979allocation}. This is true when $\alpha$ is large enough and it is then better use the forced balance approach.  When $\alpha$ is rather small, however, that test has 0 power, while the Bernoulli trial may have substantial power.
The idea that using $\mathcal{W}=\{0,1\}^n$ leads to higher-resolution \emph{p}-values, has not been mentioned before to our knowledge.
Nowadays, the use of large multiple testing corrections is more common than in the past, so higher-resolution, exact \emph{p}-values can  clearly be of interest.

Suppose we use $\mathcal{W}=\{0,1\}^n$. Then, if we happen to draw $W=(0,...,0)$ or $W=(1,...,1)$,  the value of the statistic \eqref{te2} is $0$ and we can have no hope of rejecting $H_0$ (if $\alpha=0.05$). Hence we might exclude $(0,...,0)$ and  $(1,...,1)$, and perhaps more elements, from $\mathcal{W}$. 
In any case, if $\alpha<1/N$, it can be useful to consider a design with $|\mathcal{W}|$ larger than $N$. 
Note that  in practice, we should choose $\mathcal{W}$ before administering the treatments. Once the treatments have been given, we cannot change our minds about $\mathcal{W}$.
The randomization test that uses all patterns from  $\{0,1\}^n$ except $(0,...,0)$ and  $(1,...,1)$ is further studied with simulations in Section \ref{secsim}.

\subsection{Randomization testing under a random sampling model}

For completeness we note the following, but it can be skipped at a first read.
Existing permutation tests based on random sampling from populations rely on a group structure. However, in some cases, we can use an alternative approach to avoid the requirement of a group structure also in this setting. This is a novel idea, to our knowledge (and arguably falls outside the categories of Figure \ref{venndiagram}).
The approach is analogous to the test in Section \ref{highrescc}.
Suppose that we are comparing two populations, for example, a population of cases and a population of controls, or Englishmen and Frenchmen. 
Let $W$ have the uniform distribution on $\mathcal{W}=\{0,1\}^n$ or some subset thereof, as before. 
Then we could draw from the two populations as indicated by $W$, i.e., for every $1\leq i \leq n$, the $i$-th individual is drawn from  the first population if $W_i=0$ and from the second population if $W_i=1$. 
For $1\leq i \leq n$, let $Y_i$ be the observation for the $i$-th individual, for example his or her stature.
We can then perform a test exactly as  in Section \ref{highrescc}, using the test statistic \eqref{te2} and the \emph{p}-value \eqref{pvrtest}.

If we take $\mathcal{W}$ as in \eqref{refset}, then the test will be equivalent to a standard permutation test. For many other choices of $\mathcal{W}$, we obtain a novel type of test. 
If we take for instance  $\mathcal{W}=\{0,1\}^n$, then  the number of observations drawn from each population will be random, with only the total number of observations being fixed at $n$. 
In many situations this would be impractical, for example because there is only a limited, fixed number of cases. We will not pursue such tests further here.

\section{Empirical example} \label{secsim}

Here we  illustrate the idea in Section \ref{highrescc} with a simple simulation study. We considered two randomization tests:  a standard randomization test and an alternative test that provides higher-resolution \emph{p}-values, as discussed in Section \ref{highrescc}.
The setting was as in the example in Section \ref{highrescc}, with $n=8$. Every $Y_i$ was distributed as the absolute  value of a  $N(0,1)$ variable if $W_i=0$; if $W_i=1$ it had the same distribution, but with an increase in mean of $2$. Under the null hypothesis, the distribution of $Y_i$ does not depend on $W_i$.
The first test considered was the standard randomization test. This test uses $N=(n!)/((n/2)!(n/2)!)=70$ permutations. 
The second test was based on all  relabellings in $\{0,1\}^n$ excluding $(0,...,0)$ and $(1,...,1)$. Thus, this test used $2^n-2=254$ relabellings.
We  used the test statistic \eqref{te2}.
By Theorem \ref{dtest}, both tests control the type I error rate. Moreover, the first test is exact if $\alpha\in(0,1)$ is a multiple of $1/70$. The second test is exact if $\alpha$ is multiple of $1/254$.
It is important to note that the two tests are based on  different randomization schemes, i.e., on different data gathering mechanisms. In practice the type of test should already be decided upon before running the physical experiment.

 \begin{table}[!h] \normalsize  
\begin{center} 
\caption{Performance of an alternative to the standard  randomization test. Test 1 is the standard test, based on a forced balance procedure. Test 2 is the alternative test, based on more relabellings.  } 
    \begin{tabular}{ l  l l  l  l l l l   }    
\hline \\[-0.4cm]
 &  & \multicolumn{5}{l}{\qquad \qquad  \qquad \qquad \qquad  $\alpha$} \\ \cline{3-7} 
 &  test & 1/254 ($\approx .0039 $) & .005 & .01 & .02 &   .05      \\ \hline 
  size &  test 1   \qquad \quad & 0 & 0 & 0& .0122 & .0418       \\ 
   & test 2 \qquad \quad & .0034& .0034 & .0076 & .0190 & .0464     \\       \hline
    power &  test 1  \qquad \quad & 0& 0 & 0  & .9011 & .9725  \\ 
  &  test 2  \qquad \quad & .5443& .5443& .7027 & .8436 & .9316      \\      \hline

    \end{tabular}
\label{table:2tests}
\end{center}
\end{table}

In Table \ref{table:2tests}, for different values of the significance level $\alpha$, the estimated level and power of the two tests are shown. Every estimate in the table is based on $10^4$ repeated simulations. The regular randomization test had no power for $\alpha<1/70$, due to the fact that only 70 relabellings are available with this approach. Test 2, which is based on 254  relabellings, however, did have substantial power for $1/254 \leq \alpha<1/70$, as explained in Section \ref{highrescc}. 
In the table, the estimated size for $\alpha=1/254$ is $0.0034$, which is approximately the true size $1/254\approx 0.0039$. Note that for $\alpha=0.005$, the size and power are the same as for $\alpha=1/254$. The reason is the discreteness of the \emph{p}-value: $0.005$ lies between $1/254$ and $2/254$.

\section{Discussion}
In this paper, we have distinguished between two types of permutation-based tests: 
tests that fundamentally rely on a group structure and tests based on treatment randomization, which do not necessarily require a group structure.
 We have discussed that in settings where treatments are randomly assigned, it can be useful to consider a  randomization scheme that does not correspond to a group. In particular, this allows  obtaining higher-resolution exact \emph{p}-values than are possible with standard randomization tests. 
 This paper also provides the caveat that referring to the Lady Tasting Tea experiment as an example of a permutation test can be misleading, since the reasoning underlying this experiment is not based on a group structure.

The two types of tests between which we distinguish roughly correspond to respectively ``permutation tests" and ``randomization tests" in the sense of \citet{onghena2018randomization} and \citet{rosenberger2019randomization}. As we mentioned, the use of these terms has been rather inconsistent throughout the literature. For example,  \citet[][p.1]{edgington2007randomization} write that ``\emph{randomization tests}
are a subclass of statistical tests called \emph{permutation tests}", while \citet{onghena2018randomization} proposes to use the terms for strictly distinct classes of tests.
In any case, we propose to use the term ``randomization test" only when there is some form of treatment randomization. This is in line with  \citet{kempthorne1969behaviour},  \citet{edgington2007randomization}, \citet{onghena2018randomization} and \citet{rosenberger2019randomization}.

As mentioned in the introduction, it would not be straightforward to identify the first permutation test \citep{berry2014chronicle}. In any case, it is clear that, once the concepts of randomization of treatments and random sampling from populations had been established in the 1920's \citep{rubin1990comment, fisher1925statistical, neyman1928use}, the way was paved for the theoretical development of permutation-based tests. However, until the 1980's, there  was limited interest in permutation-based  procedures, due to lack of  access to fast computers.
Nowadays, the opposite is true  \citep{albajes2019voxel, hemerik2019permutation, rao2019permutation} 
and this article discusses important differences between two types of tests involving permutations or other reassignments.

\appendix 

\section{Enforced covariate balance in a Bernoulli trial}   \label{Acovb} 
Consider the Bernoulli trial discussed in Section \ref{highrescc}. Here we illustrate how we can use this type of  approach while also enforcing covariate balance. We  provide a simple example where there is one (binary) covariate, say sex, and we want the percentage of women to be the same in the two treatment groups. Generalizations can also be formulated.

Suppose that $n>0$ is divisible by 4 and that there are $n/2$ men and $n/2$ women. The goal is to achieve covariate balance, which means that we want the same fraction of women in both treatment groups. One way to proceed is as follows. Randomly and indepently allocate treatments to the $n/2$ women. There are $2^{n/2}$ ways to do this. Let $l$ be the number of times that the women received treatment I. Then, randomly allocate treatment I to $l$ of the men. In this way, covariate balance is achieved: $l$ women and $l$ men have received treatment I and $n/2-l$ men and $n/2-l$ women have received treatment II. In each treatment group, the percentage of women is the same (although there is a small probability that one treatment group is empty).

We now compute the total number of possible treatment allocation patterns. Given $l$, there are 
$$\binom{n/2}{l}\binom{n/2}{l}$$ possible patterns. Hence the total number of possible patterns is
\begin{equation} \label{nrpatterns2}
\sum_{l=0}^{n/2} \binom{n/2}{l}\binom{n/2}{l}.
\end{equation}
Correspondingly, the smallest possible \emph{p}-value is the inverse of this number.
Note that with the common design that enforces equally large treatment groups, the total number of possible treatment allocation patterns is   $$\binom{n/2}{n/4}\binom{n/2}{n/4}.$$
This is smaller than \eqref{nrpatterns2}, which means that the smallest possible \emph{p}-value is larger than for the Bernoulli trial.

\setlength{\bibsep}{3pt plus 0.3ex}  
\def\bibfont{\small}  

\bibliographystyle{biblstyle}
\bibliography{bibliography}

\end{document}